\documentclass{article}[12pt]
\usepackage{indentfirst}
\usepackage{tabularx}
\usepackage[a4paper]{geometry}
\usepackage[utf8]{inputenc}
\usepackage[T1]{fontenc}
\usepackage{lmodern}
\usepackage{amsmath,amssymb,amsthm} \allowdisplaybreaks[1]
\usepackage{mathrsfs}
\usepackage{graphicx}
\usepackage{caption}
\usepackage{microtype}
\usepackage{latexsym}
\usepackage[colorlinks]{hyperref}
\usepackage{cleveref}
\usepackage{multirow}
\usepackage{setspace}
\usepackage{graphicx}
\usepackage{bookmark}
\usepackage{booktabs}
\usepackage{diagbox}
\usepackage{float}
\usepackage{color}   
\usepackage{bm}
\usepackage{makecell}
\usepackage{amssymb}
\usepackage{comment}
\usepackage{enumerate}
\usepackage{appendix}

\usepackage{fullpage}
\usepackage{times}
\usepackage{fancyhdr,graphicx,amsmath,amssymb}
\usepackage[ruled,vlined]{algorithm2e}

\newtheorem{thm}{Theorem}[section]
\newtheorem{lm}{Lemma}[section]
\newtheorem{prop}{Proposition}[section]
\newtheorem{cor}{Corollary}[section]
\theoremstyle{definition}

\newtheorem{re}{Remark}[section]
\newtheorem{ex}{Example}[section]

\newcommand{\F}{\mathbb{F}}
\newcommand{\Z}{\mathbb{Z}}
\newcommand{\C}{\mathcal{C}}
\newcommand{\D}{\mathcal{D}}

\newcommand{\x}{\textbf{x}}
\newcommand{\y}{\textbf{y}}

\title{Expanding self-orthogonal codes over a ring $\Z_4$ to self-dual codes and unimodular lattices }

\author{ Minjia Shi\thanks{smjwcl.good@163.com},
	Sihui Tao\thanks{taosihui2022@163.com}, Jihoon Hong\thanks{rjekfl@sogang.ac.kr}, Jon-Lark Kim\thanks{jlkim@sogang.ac.kr}
	\thanks{Minjia Shi and Sihui Tao are with the Key Laboratory of Intelligent Computing Signal Processing, Ministry of Education, School of Mathematical Sciences, Anhui University, Hefei, Anhui 230601, China. State Key Laboratory of Integrated Service Networks, Xidian University, Xi'an,
		710071, China. Jihoon Hong and Jon-Lark Kim are with the Department of Mathematics, Sogang University, Seoul, South Korea. This research (M. Shi) is supported by National Natural Science Foundation of China (12471490). This research (J.-L. Kim and J. Hong) is partially supported by the 4th BK 21 ''Nurturing team for creative and convergent mathematical science talents'' of Department of Math at Sogang University.}}

\date{}

\makeatletter

\newcommand{\Rmnum}[1]{\expandafter\@slowromancap\romannumeral #1@}
\makeatother

\begin{document}
	\maketitle
	\begin{abstract}
		Self-dual codes have been studied actively because they are connected with mathematical structures including block designs and lattices and have practical applications in quantum error-correcting codes and secret sharing schemes.
		Nevertheless, there has been less attention to construct self-dual codes from self-orthogonal codes with smaller dimensions. Hence, the main purpose of this paper is to propose a way to expand any self-orthogonal code over a ring $\Z_4$ to many self-dual codes over $\Z_4$.
We show that all self-dual codes over $\Z_4$ of lengths $4$ to $8$ can be constructed this way. Furthermore, we have found five new self-dual codes over $\Z_4$ of lengths $27, 28, 29, 33,$ and $34$ with the highest
Euclidean weight $12$. Moreover, using Construction $A$ applied to our new Euclidean-optimal self-dual codes over $\Z_4$, we have constructed a new odd extremal unimodular lattice in dimension 34 whose kissing number was not previously known.

	\end{abstract}
	
	\section{Introduction}\label{Introduction}

	A linear code is called {\em self-dual} if it is equal to its dual and {\em self-orthogonal} if it is contained in its dual. Since the beginning of Coding Theory, self-dual codes have been actively studied. Some well-known examples include the binary extended Hamming code of length 8 and the binary extended Golay code of length 24.
	Self-dual codes play an important role in Coding Theory and have  applications in various mathematical structures such as block designs~\cite{design,design1}, lattices~\cite{Lattices}, modular forms~\cite{modular forms}, and sphere packings~\cite{Sphere Packings}.
	The classification problem of self-dual codes up to equivalence is one of the difficult problems. Researchers have been interested in classifying self-dual codes with different lengths $n$ according to their equivalence \cite{classfication of sd code,Fields}.

	On the other hand, self-orthogonal codes have been less studied and more difficult to classify them because there are many self-orthogonal codes of a given length and with various dimensions. Freibert and Kim~\cite{FreKim} studied a chain of self-orthogonal subcodes of a fixed self-dual code such that the minimum distance of each self-orthogonal subcode is as large as possible, which are called {\em optimum distance profiles (ODP) of a self-dual code}. The authors of \cite{FreKim} gave the ODP of Type II self-dual codes of length up to 24 and the five extremal Type II codes of length 32, and proposed a partial result of the ODP of the extended quadratic residue code of length 48.
	
	There are other related ways to construct self-dual codes from various ways.
	One of the famous constructions is the gluing theory~\cite{PleSlo}, which produces self-dual codes by the direct sum of several self-orthogonal codes $\C_1, \C_2, ..., \C_t$ over a finite field. Kim~\cite{Kim} gave a systematic method called the building-up construction in order to construct self-dual codes from a shorter length self-dual codes.
	Recently, Kim et al.~\cite{KimKimLee} constructed good self-orthogonal codes from good linear codes by the embedding method. Then there has been active work to construct good self-orthogonal codes~\cite{KimCho},~\cite{ShiLiHelKim}.
	However, not many theories have been discussed on how to construct self-dual codes from self-orthogonal codes.

 Durğun  \cite{max s.o code} also proposed an approach for constructing a self-dual code denoted as $\C_{sd}$ from a maximal self-orthogonal code $\C_m$ over the finite field $GF(q)$.
This process specifically entails augmenting the generator matrix of $\C_m$ with an additional row and column.
Moreover,  it was shown \cite{max s.o code} that every self-orthogonal code over a finite field can be expanded to a self-dual code over the field. Hence we propose a method to construct self-dual codes from a given self-orthogonal code over a ring $\Z_4$ as a next step since linear codes over $\Z_4$ have been actively studied.

	\medskip

	In this paper, we study the opposite direction of ODP. More precisely, we start from a self-orthogonal code over $\Z_4$, augment a vector to produce a larger self-orthogonal code containing the given code, and repeat this process until we get a self-dual code over $\Z_4$ including all the intermediate self-orthogonal codes.
	This approach is theoretically interesting and connects a given self-orthogonal to various self-dual codes.

	We first propose a way to expand any self-orthogonal code $\C$ over $\Z_4$ to a self-dual code $\C_{sd}$ over $\Z_4$. Especially, whenever $\C$ is free, we can expand $\C$ into a self-dual code with type $4^{k_1}2^{k_2}$ with largest possible $k_1$.
	 Several examples listed in Sections 3 and 4 are self-dual codes over $\Z_4$ with the largest Lee or Euclidean minimum distance. Moreover, using Construction $A$ applied to our new Euclidean-optimal self-dual codes over $\Z_4$, we have constructed a new odd extremal unimodular lattice in dimension 34 whose kissing number was not previously known.
	
	This paper is organized as follows.
	In Section 2, we help readers recall some basic definitions and concepts which are used throughout the paper. In Section 3, we prove that any self-orthogonal code over $\Z_4$ can be expanded into a self-dual code over $\Z_4$ under some conditions. In Section 4, we construct new self-dual codes  over $\Z_4$ of length $27, 28, 29, 33,$ and $34$ with the highest Euclidean weight $12$ by considering expanding self-orthogonal codes over $\Z_4$. We construct odd extremal unimodular lattices in dimensions $27, 28, 29, 33,$ and $34$.
	In Section 5, we conclude the whole paper. In the appendix, we give three algorithms to expand any self-orthogonal code to many self-dual codes over $\Z_4$ with the given type.

	\section{Preliminaries}
	
	In this section, we provide basic definitions and notations. Readers can refer to
\cite{HufKimSol},~\cite{HufPle},~\cite{JoyKim},~\cite{MacWilliams}.

Let $\mathbb Z_4$ be the ring $\mathbb Z_4 / 4\mathbb Z = \{0, 1, 2, 3 \}$ and $\F_q$  the finite field with $q$ elements.
	A {\em linear code} $\C$ of length $n$ over $\mathbb Z_4$ is a $\mathbb Z_4$-submodule of $\mathbb Z_4^{n}$. The {\em dual} $\C^{\bot}$ of a linear code $\C$ of length $n$ over $\mathbb Z_4$ is defined by
	$$\C^{\bot}=\left\{ \textbf{x}\in R^{n}~|~\textbf{x}\cdot \textbf{y}=0 ~{\rm for~all}~\textbf{y}~\in \C \right\},$$
	where $\cdot$ is the Euclidean inner product.
	A linear code $\C$ is called {\em self-orthogonal} if $\C\subseteq \C^{\bot}$, and {\em self-dual} if $\C=\C^{\bot}$.
	Two linear codes $\C_1$ and $\C_2$ of length $n$ over $R$ are called {\em permutation equivalent} if there exists an $n\times n$ permutation matrix $P$ such that $\C_2=\C_1P$.
	
	 Let $\x \in \Z_4^{n}$ and $n_i(\x)$ denote the number of components of $\x$ which are equal to $i$ for each $i \in \Z_4$.
	The {\em Lee weight} of $\x$ is $\text{wt}_{L}(\x)=n_1(\x)+2n_2(\x)+n_3(\x)$. The {\em Lee distance} between two codewords $\x$ and $\y$ is defined by
	$d_L(\x,\y)=\text{wt}_L(\x-\y)$. The {\em Euclidean weight} of $\x$ is $\text{wt}_{E}(\x)=n_1(\x)+4n_2(\x)+n_3(\x)$. The {\em Euclidean distance} between two codewords $\x$ and $\y$ is defined by
	$d_E(\x,\y)=\text{wt}_E(\x-\y)$.
A self-dual code over $\Z_4$ with all Euclidean weights a multiple of $8$ is called {\it Type II}, and otherwise {\it Type I}. In this paper, we consider mainly Type I codes over $\Z_4$.
	
	 It is well known that a $\Z_4$-linear code $\C$ is permutation equivalent to a code with generator matrix $G$ of the standard form
	\begin{equation}
		G=\left[ \begin{matrix}
			I_{k_1}&		A&		B_1+2B_2\\
			O&		2I_{k_2}&		2C\\
		\end{matrix} \right]   \label{1},
	\end{equation}
	where $A$, $B_1$, $B_2$, and $C$ are matrices with entries from $\F_2$, and $O$ is the $k_2\times k_1$ zero matrix. The code $\C$ is of type $4^{k_1}2^{k_2}$.
	The {\em residue code} ${\text{Res}}(\C)$ of $\C$ is defined by
	$${\text{Res}}(\C)=\{(c_1~({\rm mod}~2),c_2~({\rm mod}~2),\ldots,c_n~({\rm mod}~2))~|~(c_1,c_2,\ldots,c_n)\in \C\}.$$ The {\em torsion code}
	${\text{Tor}}(\C)$ of $\C$ is defined by
	$${\text{Tor}}(\C)=\left\{ \boldsymbol{c}\in \F_{2}^{n}~|~2\boldsymbol{c}\in \C \right\}.$$
	If $\C$ has generator matrix $G$ in the standard form (\ref{1}), then ${\text{Res}}(\C)$ and ${\text{Tor}}(\C)$ have  generator matrices
	\begin{equation*}
		{G}_\text{Res}=\left[ \begin{matrix}
			I_{k_1}&		A&		B_1\\
		\end{matrix} \right],
		G_\text{Tor}=\left[ \begin{matrix}
			I_{k_1}&		A&		B_1\\
			\bf{0}&		I_{k_2}&		C\\
		\end{matrix} \right],
	\end{equation*} respectively.

	We explain some notations~\cite{classfication of sd code} and their meanings which will be used in the following sections.
	
	\begin{table}[htbp]
		\caption{Symbols and their meanings} 
		\label{notations}
		\centering
		\begin{tabular}{cc}
			\toprule
			notations & representations \\
			\midrule 
			$\mathscr{A} _1$ &  $\left\{ 0,2 \right\} $\\
			${v_1}$ &   $0101...01$  \\
			${v_2}$ &   $ 00...0011$ \\
			$\D_{2m}$ & the self-dual code of length $2m$ generated by
			$11130...0, 0011130...0, ..., 0...01113$\\
			
			$\D_{2m}^{\bigcirc }$ & the self-dual code generated by $\D_{2m}$ and $1300...0011$\\
			
			$\D_{2m}^{+}$ & the self-dual code generated by $\D_{2m}$ and $2v_2$ \\
			
			$\D_{2m}^{\oplus}$ &  the self-dual code generated by $\D_{2m}^{\bigcirc }$ and $\D_{2m}^{+}$ \\
			$\mathcal{E}_{7}^{+}$ & the self-dual code generated by 10003110,1010031,1101003,2222222\\
			
			
			
			
			\toprule
		\end{tabular}
	\end{table}

Self-dual codes over $\Z_4$ can be used to construct Euclidean lattices.

We define some definitions and facts on lattices.
An Euclidean lattice (simply lattice) $L$ in dimension $n$ is called {\it unimodular} if $L = L^*$, where $L^* = \{ \x \in \mathbb R^n ~|~ (\x, \y) \in \mathbb Z {\mbox{ for all }} \y \in L \}$ under the Euclidean inner product $(\x, \y)$. The {\it norm} of a vector $\x$ is $(\x, \x)$. The {\it minimum norm} of $L$ is the smallest norm among all
nonzero vectors of $L$. The kissing number $N(L)$ is the number of vectors of $L$ with minimum norm.

An unimodular lattice is {\it even} or {\it Type II} if all norms are even, otherwise {\it odd} or {\it Type I}. Let $\mu_{max}^o(n)$ denote the largest minimum norm among odd unimodular lattices in dimension $n$. An odd unimodular lattice is called {\it optimal} if it has the largest minimum norm $\mu_{max}^o(n)$. It is known that $\mu_{max}^o(n)=3$ for $n=26, 27, \dots, 30, 31, 33, 34, 35$~\cite{NebSlo}.

Let $\C$ be a Type II (Type I, respectively) code over $\Z_4$ of length $n$ and minimum Euclidean weight $d_E$ . Then the
following lattice
\[
A_4(\C) = \frac{1}{2} \{ (x_1, \dots, x_n) \in \mathbb Z^n ~|~ (x_1 \pmod{4}, \dots, x_n \pmod{4} ) \in \C \}.
\]
is an even (odd, respectively) unimodular lattice with minimum norm $\min \{4, d_E/4 \}$~\cite{BonSolBacMou}.
	
	\section{Expanding self-orthogonal codes over $\Z_4$ to self-dual codes}
	In this section, we will give basic results about expanding a self-orthogonal code to a self-dual code over $\Z_4$. We note that similar results over finite fields have  been studied in \cite{max s.o code}.

	\begin{prop} \label{prop1} (\cite[Proposition 2.3]{max s.o code})
		A self-orthogonal code $\C$ over $\mathbb F_q$ is maximal if and only if $
		\left[ \textbf{v},\textbf{v} \right] \ne0$ for all $\textbf{v}\in \C^{\bot} \setminus \C$.
	\end{prop}

In  the proof of Lemma~\ref{lm1}, Zorn's lemma was used \cite[Lemma 2.4]{max s.o code}). We give an easier proof of this without Zorn's lemma as follows.

	\begin{lm} \label{lm1} (\cite[Lemma 2.4]{max s.o code}) \label{original lemma}
		Every self-orthogonal $[n, k]$ code $\C$ over $\mathbb F_q$ is contained in a maximal self-orthogonal $[n, k' \ge k]$ code $\C_m$ such that $\C_m\subseteq \C^{\bot}$.	
	\end{lm}

	\begin{proof}
		Let $\C$ be a self-orthogonal $[n, k]$ code. Then $k \leqslant \lfloor n/2 \rfloor$. If $\C$ is maximal self-orthogonal, we are done. Assume that $\C$ is not a maximal self-orthogonal code. From Proposition
		\ref{prop1}, we know that there exists $\textbf{v}\in \C^{\bot} \setminus \C$, so that $\left[ \textbf{v},\textbf{v} \right] =0$. Let $G$ be a generator matrix of code $\C$. Take $
		G_1=\left[ \begin{array}{c}
			\textbf{v}\\
			G\\
		\end{array} \right],
		$
		and let $\C_1$ be a linear code with generator matrix $G_1$. Then we can say the code $\C_1$ is a self-orthogonal $[n, k+1]$ code and $\C\subsetneq \C_1$. If $\C_1$ is a maximal self-orthogonal code, then we are done. If not, we repeat this process (i.e., choose $\textbf{v} ' \in \C_1^{\bot} \setminus \C_1$ where $\textbf{v} ' \in \C^{\bot}$ trivially) to get a self-orthogonal code $\C_2$ such that $\C \subsetneq \C_1 \subsetneq \C_2 \subseteq \C^{\bot}$. This process will end since the dimension increases one by one and the dimension of a maximal self-orthogonal code is at most $\lfloor n/2 \rfloor$. We denote a maximal self-orthogonal code containing $\C$ obtained in this process by $\C_m$, which then satisfies $\C_m \subseteq \C^{\bot}$.	
	\end{proof}

	\begin{thm} \label{thm-ext-1}
		If $\C$ is a self-orthogonal code of length $n$ over $\Z_4$ with type $4^{k_1}2^{k_2}$, then $\C$ can be expanded to a self-dual code $\C'$ of length $n$ over $\Z_4$ with type $4^{k_1}2^{k_2'}$, where $k_2'=n-2k_1$.
	\end{thm}
	
	\begin{proof}
		Let $\C$ be a self-orthogonal code of length $n$ over $\Z_4$ with type $4^{k_1}2^{k_2}$. If $\C$ is self-dual, then we are done. Now we assume that $\C$ is not a self-dual code.
		Let ${\mbox{Res}}(\C)$ and ${\mbox{Tor}}(\C)$ be the residue code and torsion code of $\C$, respectively. Then ${\mbox{Res}}(\C)=\mathcal{A}$ is a $k_1$-dimensional doubly-even self-orthogonal binary code and ${\mbox{Tor}}(\C)=\mathcal{B}$ is also a  binary code with dimension $k_1+k_2$. Since $\C$ is not self-dual, so we can say $\mathcal{B} \subsetneq \mathcal{A}^{\bot}$. Let $H$ be the parity matrix of code $\mathcal{B}$ and $G_{\mathcal{A}^{\perp}}$ be the generator matrix of code $\mathcal{A}^{\perp}$. We can write $G_{\mathcal{A}^{\perp}}$ as
		$$
		G_{\mathcal{A}^{\perp}}=\left[ \begin{array}{l}
			~~~{\bf \alpha}_1~~~\\
			~~~{\bf \alpha}_2\\
			~~~ ~~ \vdots\\
			~~{\bf \alpha}_{n-k_1}\\
		\end{array} \right].
		$$
		For any row vector $\alpha_i$ in $G_{\mathcal{A}^{\perp}}$, we determine whether $\alpha_i$  satisfies the following equation, where $1\leqslant i\leqslant n-k_1$.
		\begin{equation} \label{aH^T=0}
			\alpha _iH^T=0
		\end{equation}
		If so, we delete the corresponding $\alpha_i$ in the matrix $G_{\mathcal{A}^{\perp}}$. Since the dimension is finite, we can find all linear independent vectors ${\bf \alpha}_{j_1}, {\bf \alpha}_{j_2}, ...,{\bf \alpha}_{j_{n-2k_1-k_2}}$ and each vector is from $\mathcal{A}^{\perp}~ \backslash ~ \mathcal{B}$.
		Assume that $G$ is the generator matrix of code $\C$. Then,
		$$
		G'=\left[ \begin{array}{l}
			~~~~~G~~~~~\\
			~~~~2{\bf \alpha}_{j_1}~~~\\
			~~~~2{\bf \alpha}_{j_2}\\
			~~~~ ~~ \vdots\\
			2{\bf \alpha}_{j_{n-2k_1-k_2}}\\
		\end{array} \right]
		$$
		generates the self-dual code $\C'$ over $\Z_4$ which contains $\C$ as a subcode. Note that $\C'$ is with type $4^{k_1}2^{k_2'}$, where $k_2'=n-2k_1$.
	\end{proof}

	In Theorem~\ref{thm-ext-1}, we only add vectors to $\C$ of the form $2{\bf c}$, where ${\bf c} \in {\mbox{Res}}(\C)^{\bot} ~\backslash ~{\mbox{Tor}}(\C)$ to get a self-dual code over $\mathbb Z_4$ containing $\C$. This method keeps the value of $k_1$. We make the Algorithm \ref{algorithm1} based on Theorem \ref*{thm-ext-1} in Appendix. One can use this algorithm to obtain a self dual code which keeps the type of $k_1$.

	\begin{thm}(\cite[Thoerem 3]{classfication of sd code}) \label{thm-z4-sd}
		If $\mathcal A$ and $\mathcal B$ are binary codes with $\mathcal A \subseteq \mathcal B$, then there is a linear code $\C$ over $\mathbb Z_4$ with ${\mbox{Res}}(\C)=\mathcal A$ and ${\mbox{Tor}}(\C)=\mathcal B$. In addition, if $\mathcal A$ is doubly-even self-orthogonal and $\mathcal B \subseteq \mathcal A^{\perp}$, then there is a self-orthogonal code $\C$ over $\mathbb Z_4$ with ${\mbox{Res}}(\C)=\mathcal A$ and
		${\mbox{Tor}}(\C)=\mathcal B$. Furthermore, if $\mathcal B = \mathcal A^{\perp}$, then $\C$ is a self-dual code over $\mathbb Z_4$.
	\end{thm}

	Let us explain the process to construct a self-dual code given in Theorem~\ref{thm-z4-sd} when $\mathcal B = \mathcal A^{\perp}$. For any binary doubly-even code $\mathcal A$, without loss of generality we can assume that its standard form of generator matrix $G_1$ is given as
	\begin{equation} \notag
		G(\mathcal A) = G_1=\left[ \begin{matrix}
			I_{k_1}&		A&		B\\
		\end{matrix} \right].
	\end{equation} Since the binary self-orthogonal code $\mathcal{B}=\mathcal{A}^\perp$ contains $\mathcal{A}$, we can assume that the standard form of generator matrix of code $\mathcal{B}$ is
	\begin{equation} \notag
		G(\mathcal B) = G_2=\left[ \begin{matrix}
			I_{k_1}&		A&		B\\
			0 & I_{k_2}  & C \\
		\end{matrix} \right].
	\end{equation}
	Take \begin{equation} \label{2}
		G=\left[ \begin{matrix}
			I_{k_1}&		A&		B\\
			0 & 2I_{k_2}  & 2C \\
		\end{matrix} \right]
	\end{equation} as a generator matrix of code $\C$ over $\Z_4$, then we have that Res$(\C)=\mathcal{A}$ and Tor$(\C)=\mathcal{B}$. Next, we need to consider the orthogonal relationship of different rows in $G$. In $\Z_4$, we cannot guarantee two different rows $i,j$ where $1\leqslant j<i\leqslant k_1$ are orthogonal. So we have to modify this generator matrix (\ref{2}). We can achieve an orthogonal relationship by replacing the $(i,j)$th entry with the inner product modulo 4 of row $i$ and row $j$. Then modified matrix $G'$ can generate a self-dual code over $\Z_4$. More detailed procedures are given in Lemma \ref{lem-binquar}.
		
	\medskip

The authors of \cite{Fields} in Section III gave a method to construct many self-dual codes over $\Z_4$ from  binary doubly-even codes.	The authors of \cite{Z_p^2} in Section 4 gave a method to
	construct  many self-dual codes over $\Z_9$ with the given length $n$. We modify this method to construct many self-dual codes over $\Z_4$. The following lemma  will give a detailed description for obtaining  many self-dual codes over $\Z_4$ from a given binary doubly-even code. We give the  proof of the following lemma in order to be self-contained.
	
	\begin{lm}\label{lem-binquar}(\cite[Lemma 5]{mass formula over Z_4},\cite[Section III]{Fields})
		Given a binary doubly-even code $\C_1$ in a standard form with dimension $k_1$ and length $n$, we can  construct $2^{\frac{k_1\left( k_1+1 \right)}{2}}$ distinct self-dual codes $\C_{sd}$ over $\Z_4$ with the same length $n$ and  ${\text{Res}(\C_{sd})}=\C_1$.
	\end{lm}

	\begin{proof}
	    Assume that $\C_1$ has a generator matrix  $G_1$ with a basis $\{ e_1, e_2, \dots, e_{k_1} \}$ so that
	    $$
	    G_1=\left[ \begin{array}{l}
	    	e_1\\
	    	e_2\\
	    	\vdots\\
	    	e_{k_1}\\
	    \end{array} \right] .
	    $$
	A  set of $k$ vectors $\left\{ e_1^*, e_2^*, \dots, e_{k_1}^*\right\}$ can be obtained from  $e_i \cdot e_j^*=\delta_{ij}$, where the Kronecker delta $\delta_{ij}$ is defined  as
	$$
	\delta _{ij}=\begin{cases}
		1~~i=j,\\
		0~~i\ne j.\\
	\end{cases}
	$$
	Let $$
	S=\left[ \begin{array}{l}
		e_1^*\\
		e_2^*\\
		\vdots\\
		e_{k_1}^*\\
	\end{array} \right].
	$$

	We choose all possible $k_1 \times k_1$ matrices $M=(m_{ij})$ over $\F_2$ such that $e_i\cdot e_j\equiv 2\left[ m_{ij}+m_{ji} \right] \pmod{4}$ for $1\leqslant i\leqslant k_1$ and $1\leqslant j\leqslant k_1$. The number of distinct matrices $M$ is $2^{\frac{k_1\left( k_1+1 \right)}{2}}$ since it suffices to count the elements on or above the diagonal. The generator matrix $G_{sd}$ of each self-dual code $\C_{sd}$ over $\Z_4$ is obtained by taking $G_1+2MS$ as the first $k_1$ rows of $G_{sd}$ and 2 times a complement of $\C_1$ in $\C_1^{\bot}$ as the last $n-2k_1$ rows of $G_{sd}$.
	\end{proof}
	
	In what follows, we show that there exists another special method that increases the dimension of residue code as in Theorems \ref{thm-ext-2} and \ref{thm-ext-3}.

	\begin{thm} \label{thm-ext-2}
		Suppose that $\C$ is a self-orthogonal code of length $n$ over $\Z_4$ with type $4^{k_1}2^{0}$ in a standard form. If there exists a binary doubly-even self-orthogonal code $\mathcal A$ with parameters $[n, k_1' \geqslant k_1]$ which contains ${\text{Res}}(\C)$ as a subcode, then $\C$ can be expanded to many self-dual codes $\C''$ of type $4^{k_1'}2^{k_2'}$, where  $k_2'=n-2k_1'$.
	\end{thm}
	
	\begin{proof}
		
		Note that $\mathcal A^{\perp}$ is a binary $[n, n-k_1']$ code such that $k_1' \leqslant  \lfloor \frac{n}{2} \rfloor$ and
		${\mbox{Res}}(\C) \subseteq \mathcal A \subseteq \mathcal A^{\perp}$.
		Without loss of generality, we can assume that  $\C$  has a generator matrix of the form
		 $
		G_0=\left[\begin{matrix}
			I_{k_1}&		A&		B_1+2B_2\\
		\end{matrix} \right],
		$ where $A$, $B_1$ and $B_2$ are binary matrices.
		Then, $G_1=
		\left[ \begin{matrix}
			I_{k_1}&		A&		B_1\\
		\end{matrix} \right]
		$  is the generator matrix of ${\mbox{Res}}(\C)$.
		With these assumptions,   the code $\mathcal{A}$ has a generator matrix $G_2$ given by
		\begin{equation*}
			G_2=\left[ \begin{matrix}
				I_{k_1}&		A&		B_1\\
				\bf{0}&		I_{k_1'-k_1}&		C\\
			\end{matrix} \right].
		\end{equation*}

		We modify the binary generator matrix $G_2$ into $G_4$ over $\mathbb Z_4$ by following two steps:
		\begin{enumerate}[(1)]
			\item Firstly, according to Lemma \ref{lem-binquar} we change the matrix $\left [\begin{matrix}
			    I_{k_1'-k_1} & C
			\end{matrix}\right]$ to $\left [\begin{matrix}
			E & C
		\end{matrix}\right]$ in a way that any two rows (not necessarily distinct) in this matrix $\left [\begin{matrix}
		E & C
	\end{matrix}\right]$ are orthogonal over $\Z_4$.
 We have $2^{\frac{\left( k_1'-k_1 \right) ^2+\left( k_1'-k_1 \right)}{2}}$ different options to have different matrices $E$.
 Then we obtain $G_3=\left[ \begin{matrix}
				I_{k_1}&		A&		B_1\\
				\bf{0}&		E&		C\\
			\end{matrix} \right].$ \\
		\item 	Secondly, let
		\begin{equation*}
			G_4=\left[ \begin{matrix}
				I_{k_1}&		A&		B_1+2B_2\\
				\bf{0}&		E&		C\\
			\end{matrix} \right].
		\end{equation*}
		According to the orthogonal relationship over $\Z_4$ between two different rows of $G_4$,  we modify $G_4$ by changing the left zero matrix
		${\bf{0}}_{\left( k_1'-k_1 \right) \times k_1}$ into the new $\Z_4$-matrix $F$.
Let  $1+k_1 \leqslant i\leqslant k_1'$ and $1 \leqslant j\leqslant k_1$. Then the $(i,j)$th entry of $F$ equals to the inner product modulo 4 of the $i$th and $j$th rows of $G_4$. It is obvious that any entry of $F$ is either $2$ or $0$. Then we obtain the following matrix $G_5$:
		\begin{equation*}
			G_5=\left[ \begin{matrix}
				I_{k_1}&		A&		B_1+2B_2\\
				F&		E&		C\\
			\end{matrix} \right].
		\end{equation*}
		\end{enumerate}
		From the results of Theorem \ref{thm-ext-1}, we can get the final self-dual code $\C''$, which contains $\C'$ as a subcode. The code $\C''$ is with type $4^{k_1'}2^{n-2k_1}$.
	\end{proof}

	Next, we will consider a special case when ${\mbox{Res}}(\C)$ is a self-orthogonal code containing the all-ones vector.
	
	\begin{lm} (\cite[Lemma 4.6]{MacWilliams}) \label{self-dual codes lemma}
		Let $n$ be divisible by $4$, and let $\mathcal C$ be a doubly-even binary self-orthogonal $[n, k]$ code containing the all-ones vector. If $n \equiv 0 \pmod{8}$, $\mathcal C$ is contained in a doubly-even binary self-dual code $\C_{sd}$. If  $n \equiv 2 \pmod{8}$, there is no doubly-even binary self-dual code $\C_{sd}$ containing $\mathcal C$.
	\end{lm}

	\begin{cor}
		Suppose that $\C$ is a self-orthogonal code of length $n$ over $\Z_4$ with type $4^{k_1}2^{0}$. If $n$ is divisible by $8$ and  ${\mbox{Res}}(\C)$ is a self-orthogonal code containing the all-ones vector, then $\C$ can be expanded to a self-dual code $\C'$ of type $4^{k_1'}2^{0}$, where  $k_1'=\frac{n}{2}$.
	\end{cor}
	
\begin{proof}
	From the result of Lemma \ref{self-dual codes lemma}, we can obtain the binary doubly-even self-dual code $\C_{sd}$.
	By applying Lemma \ref{lem-binquar}, we can obtain the self-dual code $\C_{sd}'$ over $\Z_4$.
\end{proof}
	
	The following examples are clearly illustrated using Theorem \ref{thm-ext-1} and Theorem~\ref{thm-ext-2}.
	
	\begin{ex}
		Let $n=4$. 	It is known~\cite{classfication of sd code} that there are exactly two self-dual codes over $\Z_4$ of length $4$.
 We have found both from one self-orthogonal code as follows. We start from the all-ones vector $[1 1 1 1]$ which gives a  self-orthogonal code $\C_1$ over $\Z_4$ with $4^{1} 2^{0}$ with generator matrix $G(\C_1)$.
Using the above theorems with $G(\C_1)$, we can expand it into two inequivalent self-dual codes with types $4^{1} 2^{2}$ or $4^{0} 2^{4}$ whose generator matrices are
		
		
		$$
		G(\C_1)= \left[ \begin{array}{cccc}
			1&		1&		1&		1 \end{array} \right]
		\rightarrow \left[ \begin{array}{cccc}
			1&		1&		1&		1\\
			\hline
			0&		2&		2&		0\\
			0&		0&		2&		2\\
		\end{array} \right] \cong \D_{4}^{\oplus},
		$$
		
		$$
		G(\C_1)= \left[ \begin{array}{cccc}
			2 & 0 & 0 & 0 \end{array} \right]
		\rightarrow
		\left[ \begin{array}{cccc}
			2&		0&		0&		0\\
			\hline
			0&		2&		0&		0\\
			0&		0&		2&		0\\
			0&		0&		0&		2\\
		\end{array} \right] = 2I_4.
		$$
		

\end{ex}
	
	\begin{ex}
		Let $n=5$. 	
Then the original self-orthogonal code $\C_1$ with type $4^{i} 2^{j}$ with $i+j=1$ can be expanded into two kinds of inequivalent self-dual codes over $\Z_4$ with type $4^{1} 2^3$ and $4^{0} 2^5$. We can obtain immediately the two inequivalent decomposable codes~\cite{classfication of sd code} as follows.
		$$
		G(\C_1)= \left[ \begin{array}{ccccc}
			1&		1&		1&		1 & 0 \end{array} \right]
		\rightarrow
		\left[ \begin{array}{ccccc}
			1&		1&		1&		1&		0\\
			\hline
			0&		2&		2&		0&		0\\
			0&		0&		2&		2&		0\\
			0&		0&		0&		0&		2\\
		\end{array} \right] \cong \D_{4}^{\oplus}\oplus \mathscr{A} _1,
		$$
		$$ G(\C_1)= \left[ \begin{array}{ccccc}
			2 & 0 & 0 & 0 & 0 \end{array} \right]
		\rightarrow  2I_5 .$$

	\end{ex}

	\begin{ex} \label{eg-n-6}
		Let $n=6$. It is known~\cite{classfication of sd code} that there are exactly three self-dual codes over $\Z_4$ of length $6$. Two of them must be decomposable since each is  the direct sum of a self-dual code over $\Z_4$ of length $4$ and the trivial self-dual code of length 1. Any self-orthogonal code $\C_1$ with type $4^i2^j$ where $i+j=1$ can be expanded into three inequivalent self-dual codes over $\Z_4$ with type $4^{0}2^6$, $4^{1}2^4$  and $4^{2}2^2$ as follows.
		$$  2I_6 $$,
		$$G(\C_1)=
		\left[ \begin{array}{cccccc}
			1&		1&		1&		1 &    0 &     0  \end{array} \right]
		\rightarrow
		\left[ \begin{array}{cccccc}
			1&		1&		1&		1&		0&		0\\
			\hline
			0&		2&		2&		0&		0&		0\\
			0&		0&		2&		2&		0&		0\\
			0&		0&		0&		0&		2&		0\\
			0&		0&		0&		0&		0&		2\\
		\end{array} \right] \cong \D_{4}^{\oplus}\oplus \mathscr{A}_1^{2},
		$$
		$$
		G(\C_1)=
		\left[ \begin{array}{cccccc}
			1&		0&		1&		1 &    2 &     3  \end{array} \right]
		\rightarrow
		\left[ \begin{array}{cccccc}
			1&		0&		1&		1 &    2 &     3  \\
			\hline
			2&		1&		0&		1&		1&		1\\
		\end{array} \right]
		\rightarrow
		\left[ \begin{array}{cccccc}
			1&		0&		1&		1&		2&		3\\
			2&		1&		0&		1&		1&		1\\
			\hline
			0&		0&		2&		0&		2&		2\\
			0&		0&		0&		2&		0&		2\\
		\end{array} \right] =G(\C_{sd})
		.$$
		Note that$G(\C_{sd})$ generates the self-dual code of length 6 with largest minimum Lee distance $d=4$.
		
	\end{ex}

	In Theorem \ref{thm-ext-2}, we only consider one condition that $\C$ is a free $\Z_4$-module. Next, we consider more complicated situation when $k_2 \ne 0$.
	
	\begin{thm} \label{thm-ext-3}
		Suppose that $\C$ is a self-orthogonal code of length $n$ over $\Z_4$ with type $4^{k_1}2^{k_2}$ in a standard form, where $k_2 \ne 0$. If there exists a binary  doubly-even self-orthogonal code $\C_{de}$ with dimension $k_1'>k_1$ such that ${\text{Res}}(\C) \subseteq \C_{de}$ and $\C_{de} \subseteq {\text{Tor}}(\C)^\perp$, then $\C$ can be expanded into many  self-dual codes $\C_{sd}$ with  type $4^{k_1'}2^{n-2k_1'}$, such that Res$(\C_{sd})\cong \C_{de}$.
	\end{thm}
	
	\begin{proof}
		Let $G$ be a generator matrix of the code $\C$, $G_1$ and $G_2$ be the generator matrices of Res($\C$) and $\C_{de}$, respectively. Since $\C_{de} \subseteq {\mbox{Tor}}(\C)^\perp$,
		the row vectors  belonging to the part of $\C_{de}~\backslash ~{\mbox{Res}}(\C)$ are always orthogonal to $2\bf u$ over $\Z_4$, where
		${\bf u} \in {\mbox{Tor}}(\C)~\backslash ~{\mbox{Res}}(\C)$.
		
		Since we can assume that
		$
		G=\left[ \begin{matrix}
			I_{k_1}&		A&	      	B_1+2B_2\\
			\bf{0}      &        2I_{k_2}&    2C     \\
		\end{matrix} \right],
		$ where $A$, $B_1$ and $B_2$ are binary matrices.
		Then  $G_1$ is equivalent to $
		\left[ \begin{matrix}
			I_{k_1}&		A&		B_1\\
		\end{matrix} \right],
		$
		and we can assume that $G_2$ is
		\begin{equation} \label{4}
			G_2=\left[ \begin{matrix}
				I_{k_1}&		A&		B_1\\
				\bf{0}&		I_{k_1'-k_1}&		D,\\
			\end{matrix} \right].
		\end{equation}
		where $D$ is a binary matrix.
		We modify the last  $k_1'-k_1$ rows of the generator matrix (\ref{4}) by following two steps:

		\begin{enumerate}[(1)]
			\item Firstly, according to Lemma \ref{lem-binquar} we change the matrix $\left [\begin{matrix}
				I_{k_1'-k_1} & D
			\end{matrix}\right]$ to $\left [\begin{matrix}
				E & D
			\end{matrix}\right]$ such that the any two rows (not necessarily distinct) in this matrix $\left [\begin{matrix}
				E & D
			\end{matrix}\right]$ are orthogonal over $\Z_4$.
			We have $2^{\frac{\left( k_1'-k_1 \right) ^2+\left( k_1'-k_1 \right)}{2}}$ different options to have different matrices $E$.
			Then we obtain $G_3=\left[ \begin{matrix}
				I_{k_1}&		A&		B_1\\
				\bf{0}&		E&		D\\
			\end{matrix} \right].$ \\
			\item 	Secondly, let
			\begin{equation*}
				G_4=\left[ \begin{matrix}
					I_{k_1}&		A&		B_1+2B_2\\
					\bf{0}&		E&		D\\
				\end{matrix} \right].
			\end{equation*}
Using the same argument as in the proof of Theorem~\ref{thm-ext-2}.
we obtain the following matrix $G_5$:
			\begin{equation*}
				G_5=\left[ \begin{matrix}
					I_{k_1}&		A&		B_1+2B_2\\
					F&		E&		D\\
				\end{matrix} \right].
			\end{equation*}
		\end{enumerate}
	
	If we let 	\begin{equation*}
		G_6=\left[ \begin{matrix}
			I_{k_1}&		A&		B_1+2B_2\\
			F&		E&		D\\
			\bf{0} &   2I_{k_2}&  2C \\
		\end{matrix} \right],
	\end{equation*}
	then $G_6$ can generate the self-orthogonal code $\C'$ containing $\C$. If $\C'$ is not self-dual, we can apply Theorem \ref{thm-ext-1} to get a self-dual code $\C_{sd}$.
	\end{proof}
	
	In fact, it is possible to find a  binary maximal doubly-even $\C_m$ to gain a final self-dual code with increased type for $k_1$. For any other situations, we will achieve the greatest number for $k_1'$  with considering a binary maximal doubly-even code $\C_m$ containing the original code $\C$. Then we will have following corollary.
	
	\begin{cor}
			Suppose that $\C$ is a self-orthogonal code of length $n$ over $\Z_4$ with type $4^{k_1}2^{k_2}$ in a standard form, where $k_2 \ne 0$. If there exists a binary maximal doubly-even self-orthogonal code $\C_m$ with dimension $k_1'>k$ such that ${\text{Res}}(\C) \subseteq \C_m$ and $\C_m \subseteq {\text{Tor}}(\C)^\perp$, then $\C$ can be expanded into many self-dual codes $\C_{sd}$ with type $4^{k_1'}2^{n-2k_1'}$ with largest dimension for residue codes, such that Res$(\C_{sd})\cong \C_m$
	\end{cor}
	
	\begin{proof}
		The proof follows from the definition of binary maximal doubly-even code $\C_m$.
	\end{proof}
	



\medskip

From now on, we consider lengths $n=6, 7, 8$. In what follows, from Examples  \ref{example-n=6} to \ref{example-n=8-k8}, we have succeeded in constructing several self-dual codes beginning with a self-orthogonal code $\C$ over $\Z_4$ generated by one vector over $\Z_4$. Interestingly, these self-dual codes are equivalent to one of the indecomposable self-dual codes mentioned in \cite{classfication of sd code}. In this sense, our work recovers all indecomposable self-dual codes from a given self-orthogonal code $\C$  with length $n=6,7,8$  over $\Z_4$. The following table describes all the known self-dual codes of lengths 6 to 8 using our theorems.
\begin{table}[htbp]
	\caption{ Summary of all indecomposable self-dual codes over $\Z_4$ with length $n=6,7,8$ based on our algorithms} 
	\label{notations}
	\centering
	\begin{tabular}{ccccc}
		\toprule
		Codes &  Length &  Type  &Starting vectors \\
		\midrule 
		$\D_{6}^{\oplus}$ &  6 & $4^22^2 $  &  111300 \\
		\midrule
		$\mathcal{E}_{7}^{+}$ & 7 &  $4^32^1$ &  1003110 \\
		\midrule
	\makecell[c]{$\mathcal{O}_8$ \\  $\mathcal{E}_8$} &8 & $4^4$ & 10111200, 01110320  \\
	\midrule
		$\mathcal{L}_8$ & 8 &  $4^32^2$ &  10101320, 01113020 \\
		\midrule
		$\D_{8}^{\oplus}$ &8 &  $4^32^2$ & 10113000\\
		\midrule
		$\mathcal{K}'_8$ &8 &  $4^22^4$  & 10111002 \\
		\midrule
		$\mathcal{K}_8$ &8 & $4^12^6$ & 11111111 \\
		\toprule
	\end{tabular}
\end{table}

	\begin{ex}\label{example-n=6}
		There is only one indecomposable self-dual code with length $n=6$ over $\Z_4$ named as $\D_{6}^{\oplus}$ with type $4^2 2^2$ from  \cite{classfication of sd code}. We aim to consider constructing it from a self-orthogonal code $\C$ obtained from the first row of the generator matrix of $\D_{6}^{\oplus}$. Hence this example has a form different from the indecomposable form in Example~\ref{eg-n-6} although two codes are equivalent.

		We start from a self-orthogonal code $\mathcal{C}$ of length $n=6$ over $\Z_4$ with generator matrix in a standard form:
		$$G(\C)= \left[ \begin{array}{ccccccc}
			1&		1&		1&		3&   0 &     0      \\
		\end{array} \right].$$
		There is a binary maximal doubly-even code $\C_m$ such that Res($\C$) $\subseteq  \C_m$ with generator matrix $G_2$, $$
		G_2 =\left[ \begin{matrix}
			1&		1&		1&		1&		0&		0\\
			0&		1&		1&		0&		1&		1\\
		\end{matrix} \right].$$
According to step 4 in Algorithm \ref{algorithm2}, we obtain $G_3$ as
$$
G_3 =\left[ \begin{matrix}
	1&		1&		1&		3&		0&		0\\
	2&		1&		1&		0&		1&		1\\
\end{matrix} \right].$$
  The matrix $G_3$ generates a self-orthogonal code $\C'$ containing $\C$. By applying Theorem \ref{thm-ext-1} with the help of Magma \cite{magma}, we can obtain two codewords $\textbf{v}_1=[2,0,2,0,2,0]$ and $\textbf{v}_2=[0,0,0,0,2,2]$ such that they are contained in $\text{Res}(\C')^\bot~\backslash~\text{Res}(\C')$.
According to step 7 in Algorithm \ref{algorithm2},  we obtain
  $$
  G_4 =\left[ \begin{matrix}
  	1&		1&		1&		3&		0&		0\\
  	2&		1&		1&		0&		1&		1\\
  	2&		0&		2&		0&		2&		0\\
  	0&		0&		0&		0&		2&		2\\
  \end{matrix} \right].$$
The matrix $G_4$ generates the self-dual code $\C_{sd}$ which contains $\C$ as a subcode with the minimum Lee distance $d=4$. The symmetric weight enumerator for $\C_{sd}$ is $a^6 + 3a^4c^2 + 8a^3c^3 + 12a^2b^4 + 3a^2c^4 + 24ab^4c + 12b^4c^2 +c^6$, which is the same as $\D_{6}^{\oplus}$, so we can say $\C_{sd}$ is equivalent to $\D_{6}^{\oplus}$.
		\end{ex}

		\begin{ex}\label{example-n=7}
There is only one indecomposable self-dual code with length $n=7$ over $\Z_4$ named as $\mathcal{E}_{7}^{+}$ with type $4^3 2^1$ from  \cite{classfication of sd code}. We aim to construct it from a self-orthogonal code $\C$ obtained from the first row of the generator matrix of $\mathcal{E}_{7}^{+}$.
	We start from the self-orthogonal code $\C$ of length $n=8$ over $\Z_4$ with generator matrix $G(\C)$ in a standard form:
$$G(\C)= \left[ \begin{array}{cccccccc}
	1&		0&		0&		3&    1 &     1     & 0   \\
\end{array} \right].$$			
There is a binary doubly-even code $\C_{de}$ such that 	Res($\C$) $\subseteq  \C_{de}$  with generator matrix $G_2$,
$$G_2= \left[ \begin{array}{ccccccccc}
	1&		0&		0&		1&    1 &     1    & 0    \\
	0&		1&		0&		1&    1 &	  0    & 1   \\
	0&      0&      1&      0&    1 &     1   & 1    \\
\end{array} \right].$$	
By using the step 4 in Algorithm \ref{algorithm2}, we choose matrix 	\begin{equation*}
	M=\left[\begin{matrix}
		0 & 0 \\
		1 & 0  \\
	\end{matrix}\right]
\end{equation*}
to obtain the matrix $G_3$ as
$$G_3=\left[ \begin{array}{ccccccccc}
	1&		0&		0&		3&    1 &     1    & 0   \\
	0&		1&		0&		1&    1 &	  0    & 1    \\
	2&      2&      1&      0&    1 &     1   & 1    \\
\end{array} \right].$$
The matrix $G_3$ generates a self-orthogonal code $\C'$ containing $\C$. By applying Theorem \ref{thm-ext-1} with Magma \cite{magma}, we can obtain one codeword $\textbf{v}=[1,1,1,1,1,1,1,1]$ such that it is contained in $\text{Res}(\C')^\bot~\backslash~\text{Res}(\C')$.
According to step 7 in Algorithm \ref{algorithm2},  we obtain
$$G_4=\left[ \begin{array}{ccccccccc}
1&		0&		0&		3&    1 &     1    & 0   \\
0&		1&		0&		1&    1 &	  0    & 1    \\
2&      2&      1&      0&    1 &     1   & 1    \\
2&      2&      2&      2&    2 &     2   & 2 \\
\end{array} \right].$$
The $G_4$ generates a self-dual code $\C_{sd}$ over $\Z_4$
with Lee minimum distance $d=4$. The symmetric weight enumerator for $\C_{sd}$ is $a^7 + 7a^4c^3 + 14a^3b^4 + 7a^3c^4 + 42a^2b^4c + 42ab^4c^2 +14b^4c^3 + c^7$, which is the same as $\mathcal{E}_{7}^{+}$, so we can say $\C_{sd}$ is equivalent to $\mathcal{E}_{7}^{+}$.

		\end{ex}

	\begin{ex}\label{example-n=8-L8}
		We start from the self-orthogonal code $\C$ consisting of the first two rows of $\mathcal{L}_8$ in \cite{classfication of sd code}.
 We want to show that we can expand $\C$ into a self-dual code $\C_{sd}$ which is equivalent to $\mathcal{L}_8$ with type $4^3 2^2$.
		The self-orthogonal code $\C$ of length $n=8$ over $\Z_4$  has generator matrix $G$ in a standard form:
$$G(\C)= \left[ \begin{array}{ccccccccc}
1&		0&		1&		0&    1 &     3    & 2 & 0   \\
0&		1&		1&		1&    3 &	  0    & 2 & 0   \\
\end{array} \right].$$
There is a binary doubly-even code $\C_{de}$ such that 	Res($\C$) $\subseteq  \C_{de}$  with generator matrix $G_2$,
$$G_2= \left[ \begin{array}{ccccccccc}
	1&		0&		1&		0&    1 &     1    & 0 & 0   \\
	0&		1&		1&		1&    1 &	  0    & 0 & 0   \\
	0&      0&      1&      0&    1 &     0   & 1 & 1   \\
\end{array} \right].$$
By using the step 4 in Algorithm \ref{algorithm2}, we choose matrix 	\begin{equation*}
	M=\left[\begin{matrix}
		0 & 0 \\
		1 & 0  \\
	\end{matrix}\right]
\end{equation*}
to obtain the matrix $G_3$ as
	$$G_3=\left[ \begin{array}{ccccccccc}
	1&		0&		1&		0&    1 &     3    & 2 & 0   \\
	0&		1&		1&		1&    3 &	  0    & 2 & 0   \\
	0&      2&      1&      0&    1 &     0   & 1 & 1   \\
\end{array} \right].$$
The matrix $G_3$ generates a self-orthogonal code $\C'$ containing $\C$. By applying Theorem \ref{thm-ext-1} with Magma \cite{magma},  we can obtain two codewords $\textbf{v}_1=[0,0,1,0,1,0,0,0]$ and $\textbf{v}_2=[0,0,0,1,1,1,1,0]$ such that they are all contained in $\text{Res}(\C')^\bot~\backslash~\text{Res}(\C')$.
According to step 7 in Algorithm \ref{algorithm2}, we obtain
$$G_4=\left[ \begin{array}{ccccccccc}
	1&		0&		1&		0&    1 &     3    & 2 & 0   \\
	0&		1&		1&		1&    3 &	  0    & 2 & 0   \\
	0&      2&      1&      0&    1 &     0   & 1 & 1   \\
	0&      0&      2&      0&    2 &     0    & 0 & 0   \\
	0&      0&      0&      2&    2 &     2     & 2 & 0 \\
\end{array} \right].$$
The $G_4$ generates a self-dual code $\C_{sd}$ over $\Z_4$
with Lee minimum distance $d=4$. The symmetric weight enumerator for $\C_{sd}$ is $a^8 + 4a^6c^2 + 22a^4c^4 + 96a^3b^4c + 4a^2c^6 + 96ab^4c^3 + 32b^8+ c^8$, which is the same as $\mathcal{L}_8$, so $\C_{sd}$ is equivalent to $\mathcal{L}_8$.
	\end{ex}

		\begin{ex}\label{example-n=8-O8E8}
			There are exactly two indecomposable free self-dual codes over $\Z_4$ of length $n=8$, i.e., $\mathcal{O}_8$ and $\mathcal{E}_8$ from \cite{classfication of sd code}. We want to construct them from one self-orthogonal code $\C$ by applying Theorem \ref{thm-ext-2}.
			We start from the self-orthogonal code  $\mathcal{C}$ of length $n=8$ over $\Z_4$ with the below generator matrix $G(\C)$ in a standard form which is easy to construct.
			$$G(\C)= \left[ \begin{array}{ccccccccc}
	1&		0&		1&		1&    1 &     2    & 0 & 0   \\
	0&		1&		1&		1&    0 &	  3    & 2 & 0   \\
			\end{array} \right].$$
			Then, we have	$$G({\mbox{Res}}(\C))=\left[ \begin{array}{ccccccccc}
	1&		0&		1&		1&    1 &     0    & 0	&  0   \\	0&		1&		1&		1&    0 &	  1    & 0  &  0    \\
			\end{array} \right].$$
			By Lemma \ref{original lemma}, Res($\C$) can be expanded into a maximal doubly-even self-orthogonal code $\mathcal{C}_m$ over $\F_2$. Here we can assume generator matrix of $\mathcal{C}_m$ as:
			$$G(\mathcal{C}_m)=\left[ \begin{array}{cccccccc}
	1&		0&		1&		1&    1 &     0    & 0  & 0   \\
	0&		1&		1&		1&    0 &	  1    & 0  & 0    \\
	0&      0&      1&      0&    1 &     1    & 0  & 1  \\
	0&      0&      0&      1&    1 &     1    & 1  & 0  \\
			\end{array} \right].$$
From the step 4 in Algorithm \ref{algorithm2}, we can obtain $8$ suitable matrices $M=[m_{ij}]$. We pick two different matrix $M_1$ and $M_2$ as follows
			\begin{equation*}
				M_1=\left[\begin{matrix}
					1 & 0 \\
					1 & 1  \\
				\end{matrix}\right],
			M_2=\left[\begin{matrix}
				1 & 1 \\
				0 & 1  \\
			\end{matrix}\right].
			\end{equation*}
Then, by Lemma \ref{lem-binquar}, we obtain $G_1$ and $G_2$, respectively as follows:
			 \begin{equation*}
			 	G_1=\left[ \begin{matrix}
1&		0&		1&		1&    1 &     2    & 0  & 0  \\
0&		1&		1&		1&    0 &	  3    & 2  & 0  \\
2&      2&     	3&		0&	  1 &	  1    & 0  & 1 \\
0&      2&      2&      3&    1 &     1    & 1  &0  \\
			 	\end{matrix} \right],
			 	G_2=\left[ \begin{matrix}
1&		0&		1&		1&    1 &     2    & 0  & 0  \\
0&		1&		1&		1&    0 &	  3    & 2  & 0  \\
0&      0&     	3&		2&	  1 &	  1    & 0  & 1 \\
2&      0&      0&      3&    1 &     1    & 1  &0  \\
			 	\end{matrix} \right].
			 \end{equation*}	   		
The matrix $G_1$ generates a self-dual code $\C_{sd1}$ over $\Z_4$ with the minimum Lee distance $d=4$. The symmetric weight enumerator for $\C_{sd1}$ is   $a^8 + 16a^4b^4 + 14a^4c^4 + 48a^3b^4c + 96a^2b^4c^2 + 48ab^4c^3 +
16b^8 + 16b^4c^4 + c^8$, which is the same as $\mathcal{E}_8$, so $\C_{sd1}$  is equivalent to $\mathcal{E}_8$. The matrix $G_2$ generates a self-dual code $\C_{sd2}$ over $\Z_4$ with the minimum Lee distance $d=6$. The symmetric weight enumerator for $\C_{sd2}$ is $a^8 + 14a^4c^4 + 112a^3b^4c + 112ab^4c^3 + 16b^8 + c^8$, which is the same as $\mathcal{O}_8$, so $\C_{sd2}$  is equivalent to $\mathcal{O}_8$.

\end{ex}

		\begin{ex}\label{example-n=8-d8+}
			We start from the self-orthogonal code $\C$ with generator matrix consisting of the first row of generator matrix of $\D_{8}^{\oplus}$ from \cite{classfication of sd code}.
  We want to show that we can expand $\C$ into a self-dual code $\C_{sd}$ which is equivalent to $\D_{8}^{\oplus}$ with type $4^3 2^2$.
The self-orthogonal code  $\mathcal{C}$ of length $n=8$ over $\Z_4$ has the below generator matrix $G(\C)$ in a standard form.
			$$G(\C)= \left[ \begin{array}{ccccccccc}
	1&		0&		1&		1&    3 &     0    & 0 & 0   \\
			\end{array} \right].$$
There is a binary doubly-even code $\C_{de}$ such that 	Res($\C$) $\subseteq  \C_{de}$  with generator matrix $G_2$,
		$$G_2= \left[ \begin{array}{ccccccccc}
1&		0&		1&		1&    1 &     0    & 0 & 0   \\
0&		1&		0&		0&    0 &	  1    & 1 & 1   \\
0&      0&      1&      0&    1 &     1    & 1 & 0   \\
		\end{array} \right].$$
By applying Lemma \ref{lem-binquar}, we choose $
M=\left[ \begin{matrix}
	0&		1\\
	0&		0\\
\end{matrix} \right]
$ to
obtain $G_3$ as
$$G_3= \left[ \begin{array}{ccccccccc}
	1&		0&		1&		1&    3 &     0    & 0 & 0   \\
	2&		1&		2&		0&    0 &	  1    & 1 & 1   \\
	0&      0&      1&      0&    1 &     1    & 1 & 0   \\
\end{array} \right].	$$		
The matrix $G_3$ generates a self-orthogonal code $\C'$ which contains $\C$ as a subcode.  By applying Theorem \ref{thm-ext-1} using Magma \cite{magma},  we can obtain two codewords $\textbf{v}_1=[0,1,0,0,1,0,0,1]$ and $\textbf{v}_2=[0,0,1,1,0,1,0,1]$ such that they are all contained in $\text{Res}(\C')^\bot~\backslash~\text{Res}(\C')$.
According to step 7 in Algorithm \ref{algorithm2}, we obtain
$$G_4= \left[ \begin{array}{ccccccccc}
	1&		0&		1&		1&    3 &     0    & 0 & 0   \\
	2&		1&		2&		0&    0 &	  1    & 1 & 1   \\
	0&      0&      1&      0&    1 &     1    & 1 & 0   \\
    0&		2&		0&		0&    0 &	  0    & 0 & 2   \\
    0&      0&      2&      2&    0 &     2    & 0 & 2    \\
\end{array} \right].$$		
The matrix $G_4$ generates a self-dual code $\C_{sd}$ with  minimum Lee distance $d=4$. The symmetric weight enumerator for $\C_{sd}$ is   $a^8 + 4a^6c^2 + 16a^4b^4 + 22a^4c^4 + 32a^3b^4c + 96a^2b^4c^2 +4a^2c^6 + 32ab^4c^3 + 32b^8 + 16b^4c^4 + c^8$, which is the same as $\D_{8}^{\oplus}$, so $\C_{sd}$ is equivalent to $\D_{8}^{\oplus}$.		
		\end{ex}

\begin{ex}\label{example-n=8-k8'}
		As before we start from the self-orthogonal code $\C$ with generator matrix consisting of the first row of the generator matrix of $\mathcal{K}'_8$~\cite{classfication of sd code}.  We will show that we can expand $\C$ into a self-dual code $\C_{sd}$ which is equivalent to $\mathcal{K}'_8$ with type $^2 2^4$.
	The self-orthogonal code $\mathcal{C}$ of length $n=8$ over $\Z_4$ has the below generator matrix $G(\C)$ in a standard form:
	$$G(\C)= \left[ \begin{array}{ccccccccc}
		1&		0&		1&		1&    1 &     0    & 0 & 2   \\
	\end{array} \right].$$
Using a similar procedure as in Example~\ref{example-n=8-d8+}, we obtain
$$G_4= \left[ \begin{array}{ccccccccc}
	1&		0&		1&		1&    1 &     0    & 0 & 2   \\
2&		1&		0&		0&    0 &	  1    & 1 & 1   \\
0&		2&		0&		0&    0 &	  0    & 0 & 2   \\
0&      0&      2&      0&    2 &     0    & 0 & 0    \\
0&      0&      0&      2&    2 &     0    & 0 & 0  \\
0&      0&      0&      0&    0 &     0    & 2 & 2  \\
\end{array} \right].$$	
The matrix $G_4$ can generate a self-dual code $\C_{sd}$ with minimum Lee distance $d=4$.  The symmetric weight enumerator for $\C_{sd}$ is   $a^8 + 12a^6c^2 + 38a^4c^4 + 64a^3b^4c + 12a^2c^6 + 64ab^4c^3 +64b^8 + c^8$, which is the same as $\mathcal{K}'_8$, so $\C_{sd}$ is equivalent to $\mathcal{K}'_8$.
\end{ex}		

\begin{ex}\label{example-n=8-k8}
	We start from the self-orthogonal code $\C$ obtained from the generator matrix consisting of the first row of the generator matrix of $\mathcal{K}_8$ \cite{classfication of sd code}. We will show that $\C$ can be expanded into a self-dual code $\C_{sd}$ which is equivalent to $\mathcal{K}_8$ with type $4^1 2^6$.
	$$G(\C)= \left[ \begin{array}{ccccccccc}
	1&		1&		1&		1&    1 &     1    & 1 & 1   \\
\end{array} \right].$$
	By applying Theorem \ref{thm-ext-1}, we can easily obtain six codewords $\textbf{v}_1=[0,1,0,0,0,0,0,1]$,$\textbf{v}_2=[0,0,1,0,0,0,0,1]
	$, $\textbf{v}_3=[0,0,0,1,0,0,0,1]$, $\textbf{v}_4=[0,0,0,0,1,0,0,1]$, $\textbf{v}_5=[0,0,0,0,0,1,0,1]$ and $\textbf{v}_6=[0,0,0,0,0,0,1,1]$ such that they are all contained in $\text{Res}(\C')^\bot~\backslash~\text{Res}(\C')$.
	According to step 7 in Algorithm \ref{algorithm2}, we obtain
$$G_1= \left[ \begin{array}{ccccccccc}
	1&		1&		1&		1&    1 &     1    & 1 & 1   \\
	0&		2&		0&		0&    0 &	  0    & 0 & 2   \\
	0&		0&		2&		0&    0 &	  0    & 0 & 2   \\
	0&      0&      0&      2&    0 &     0    & 0 & 2    \\
	0&      0&      0&      0&    2 &     0    & 0 & 2  \\
	0&      0&      0&      0&    0 &     2    & 0 & 2  \\		0&      0&      0&      0&    0 &     0    & 2 & 2  \\
\end{array} \right].$$	
The matrix $G_1$ generates  a self-dual code $\C_{sd}$ with minimum Lee distance $d=4$ which is the same as $\mathcal{K}_8$.
\end{ex}

		Next, we will consider in an ideal situation, the largest value of $k_1'$ in Theorem~\ref{thm-ext-2} and we will give the exact result in Theorem \ref{maximal dimension}.

		\begin{lm}(\cite[Thoerem 1]{Pless}) \label{v}
			Let $V$ be a finite geometry over $GF(q)$ of
			dimension $n$ with basis. Let $v$ be the dimension of a maximal self-orthogonal subspace for $n$ even. Then $v = n/2$ whenever $\left( -1 \right) ^{n/2}$ is a square in $GF(q)$, and $v=n/2-1$ whenever
			$\left( -1 \right) ^{n/2}$ is not a square in $GF(q)$. In case $n$ is odd, $v=\left( n-1 \right) /2$.
		\end{lm}
		
		\begin{cor} \label{F2}
			Let $V$ be a finite geometry over $GF(2)$ of
			dimension $n$ with basis. Let $v$ be the dimension of a maximal self-orthogonal subspace. Then for $n$ even, $v = n/2$ and for $n$ odd, $v=\left( n-1 \right) /2$.
			
		\end{cor}
		
		Combining Lemma  \ref{original lemma} and Corollary \ref{F2}, we can naturally obtain the following corollary.
		
		\begin{cor} \label{binary maximal self-orthogonal}
			Every binary self-orthogonal $[n, k]$ code $\C$  of length $n$ is contained in a maximal self-orthogonal code $\C_{max}$ of dimension $n/2$ if $n$ is even and of dimension $\left( n-1 \right) /2$ if $n$ is odd.


		\end{cor}
		\begin{proof}
By Lemma~\ref{original lemma}, $\C$ is contained in a maximal self-orthogonal $[n, k' \ge k]$ code $\C_m$ such that $\C_m\subseteq \C^{\bot}$ as a partial ordering of the chain of self-orthogonal codes containing $\C$. If $\C_m$ itself is not a maximal self-orthogonal code, then there exists a maximal self-orthogonal code $\C_{m'}$ properly containing $\C_m$. Then $\C \subseteq \C_m \subsetneq \C_{m'}$. This is a contradiction since $\C_m$ was a maximal self-orthogonal code containing $\C$. Therefore, $\C_m$ is in fact a maximal self-orthogonal code $\C_{max}$ whose dimension follows from Corollary \ref{F2}.	
		\end{proof}
		
		\begin{thm} (\cite[Thoerem 1.4.6]{HufPle}) \label{doubly-even}
			Let $\C$ be an $[n, k]$ self-orthogonal binary code. Let $\C_0$ be the set of codewords in $\C$ whose weights are divisible by four. Then either:\\
			(i) $\C$ = $\C_0$, or\\
			(ii) $\C_0$ is an $\left[ n,k-1 \right]$
			subcode of $\C$ and $\C=\C_0\cup \C_1$, where $\C_1 = \textbf{x} + C_0$ for any codeword \textbf{x} whose weight is even but not divisible by four. Furthermore, $\C_1$ consists of all codewords of $\C$ whose weights are not divisible by four.
			
		\end{thm}
		
		Since binary doubly-even self-dual codes of length $n$ exist if and only if
		$n\equiv 0 \pmod {8}$~\cite{MacWilliams}, the following statement follows from Corollary \ref{F2} and Theorem \ref{doubly-even}.

		\begin{cor} \label{doubly-even self-orthogonal code with certain dimension}
		If the dimension of a binary maximal doubly-even self-orthogonal code $\C_m$ of length $n$ is $k$, then we have
			$$
			k=\begin{cases}
				\frac{n}{2}-1~~~~~~~~~~~~~~~~~~~~{\rm{if}}~n\equiv 2,4,6 \pmod{8},\\
				\frac{n}{2}~{\rm{or}}~\frac{n}{2}-1~~~~~~~~~~~~~~{\rm{if}}~ n\equiv 0 \pmod {8},\\
				\frac{n-1}{2} ~{\rm{or}}~ \frac{n-3}{2}~~~~~~~~~~~~~ {\rm{if}}~n~{\rm{is~odd}}.
			\end{cases}
			$$
			
		\end{cor}

		\begin{cor} \label{cor4}
			Suppose that $\C$ is a self-orthogonal code of length $n$ over $\Z_4$ with type $4^{k_1}2^{k_2}$, then there exists a binary maximal  doubly-even self-orthogonal code $\C_m$ with parameters $[n, k_1' \ge k_1]$ which contains ${\mbox{Res}}(\C)$ as a subcode. When $n$ is even,
			the dimension of $\C_m$ is $\frac{n}{2} $ or $ \frac{n}{2} -1$. When $n$ is odd,  $\C_m$ has dimension $ \frac{n-1}{2} $ {\rm{or}} $\frac{n-3}{2}$.
		\end{cor}
		
		\begin{proof}
			Since $\C$ is self-orthogonal, then Res$(\C)$ is a binary doubly-even self-orthogonal code. The result followes from Corollary \ref{doubly-even self-orthogonal code with certain dimension}.
		\end{proof}
		
		\begin{thm} \label{maximal dimension}
			When $n$ is even, we can choose $k_1'$ in Theorem~\ref{thm-ext-2} at most $\frac{n}{2}$. When $n$ is odd, we can choose $k_1'$ at most $ \frac{n-1}{2} $.
		\end{thm}

		\begin{proof}
			Let $\C_m$ be a maximal doubly-even self-orthogonal code containing Res$(\C)$ described in Corollary \ref{cor4}. Then dim$(\C_m)$ is at most $\frac{n}{2}~{\rm{or}}~  \frac{n-1}{2}$ depending on the parity of $n$.
		\end{proof}

\section{New self-dual codes over $\Z_4$ of lengths $27, 28, 29, 33,$ and $34$}

In this section, we construct new self-dual codes $\Z_4$ of lengths $27, 28, 29, 33,$ and $34$ with the highest Euclidean weights. These results are achieved by applying the Algorithm \ref{algorithm1} and Algorithm
\ref{algorithm2}.

  In \cite{Ha}, the author gave some optimal self-dual codes with large lengths $n$ from $26$ to $45$ (not consecutive) over $\Z_4$ with respect to  minimum Euclidean weight. We apply our algorithms to some of them and find  five new optimal self-dual codes with respect to the Euclidean weight. 
  We show the parameters of the new codes  and distinguish our codes with the codes in~\cite{Ha} by different Lee weight distribution from Magma~\cite{magma} in Table \ref{new optimal self-dual codes}. We display their generator matrices given below Table \ref{new optimal self-dual codes}.  Those new self-dual codes are obtained by the following procedure:
  \begin{enumerate}
  	\item [(i)] Consider the submatrix $G$ consisting of the first $k_1-2$ rows of the generator matrix~\cite[Figure 1]{Ha} with a given length $n$. Let $\C_1$ be the self-orthogonal code generated by this matrix.
  	\item [(ii)] Expand $\C_1$ into many self-dual codes $\C_2$ such that Res$(\C_1) \subsetneq$ Res$(\C_2)$ and Res$(\C_2)=k_1$ by adding two more rows to $G$. For example,  the two rows between the two horizontal lines in $G_{27}^4$ are chosen carefully.
  	\item [(iii)] Expand every $\C_2$ into many self-dual codes until find an optimal one.
  \end{enumerate}

In Table \ref{new optimal self-dual codes}, $d_E$ and $d_L$ denote the minimum Euclidean weight and minimum Lee weight of the code, respectively. The notation $A_i^L$ is the number of the codewords with Lee weight $i$. The notation $A_i^E$ is the number of the codewords with Euclidean weight $i$.

  \begin{table}[H]
  \renewcommand{\arraystretch}{1.2}
  	\caption{Five new self-dual codes over $\Z_4$ with the highest Euclidean weight $12$} 
  	\label{new optimal self-dual codes}
  	\centering
  	\begin{tabular}{c|c|c|c|c|cc|c|cccccc}
  		\toprule
  $n$ &		Codes & Type &  Generator matrices &  $d_E$  & $A_{12}^E$ & $A_{16}^E$  &  $d_L$ &  $A_6^L$ & $A_8^L$  & $A_{10}^L$ & $A_{12}^L$  \\
  \hline
 	\multirow{1}{*}{27} &  $\C_{27}^4$  &  $4^72^{13}$  &$G_{27}^4$ &$12$  &$2509$  &
 	$60366$ &$6$   &  $13$  &   $142$  &  $752$   &  $5488$    \\ \cline{1-12}
  \multirow{1}{*}{28} &$\C_{28}^4$ &  $4^72^{14}$ & $G_{28}^4$ & $12$ &  $2240$ & $64827$ & $8$  &   $0$  & $315$ & $0$   & $8288$   \\  \cline{1-12}
 \multirow{1}{*}{29} &$\C_{29}^4$ &  $4^72^{15}$  & $G_{29}^4$ & $12$    &    $1716$
 & 63342 &  $6$     &  $20$ & $206$ & $861$ & $5580$  \\ \cline{1-12}
 \multirow{1}{*}{33}& $\C_{33}^4$ & $4^92^{15}$   & $G_{33}^4$ & $12$ &$625$  & $50322$  &  $6$     &  $9$  & $74$  & $480$ & $2897$ \\ \cline{1-12}
  \multirow{1}{*}{34}&$\C_{34}^4$ & $4^{10}2^{14}$    & $G_{34}^4$ & $12$ & $515$   &$45771$ &  $6$     &  $3$  & $43$  & $294$ & $1929$ \\
  		\toprule
  	\end{tabular}
  \end{table}
		
To justify that our codes in Table~\ref{new optimal self-dual codes} are new, we have compared them with the codes in \cite[Table 2]{Ha} as follows. The code with type $4^72^{13}$ and generator matrix  $M_{27}$ has $d_E =12$, $d_L=6$, $[A_6^L, A_8^L, A_{10}^L, A_{12}^L]=[5, 150, 720, 4944]$ and
$[A_{12}^E, A_{16}^E]=[2629, 59478]$.
The code with type $4^72^{14}$ and generator matrix $M_{28}$ has $d_E =12$, $d_L=8$,
       $[A_6^L, A_8^L, A_{10}^L, A_{12}^L]$$=[0, 315, 0,  7776]$ and
$[A_{12}^E, A_{16}^E]=[1728, 64827]$.
The code with type $4^72^{15}$ and generator matrix $M_{29}$ has $d_E =12$, $d_L=6$, $[A_6^L, A_8^L, A_{10}^L, A_{12}^L]=[24, 178, 937, 5464]$	and
$[A_{12}^E, A_{16}^E]=[1688, 62770]$.
The code with type $4^92^{15}$ and generator matrix $M_{33}$ has $d_E =12$, $d_L=6$, $[A_6^L, A_8^L, A_{10}^L, A_{12}^L]=[9, 66, 456, 2865]$ and
$[A_{12}^E, A_{16}^E]=[689, 50058]$.
The code with type $4^{10}2^{14}$ and generator matrix $M_{34}$ has $d_E =12$, $d_L=6$,  $[A_6^L, A_8^L, A_{10}^L, A_{12}^L]=[1, 39, 282, 1797]$ and
$[A_{12}^E, A_{16}^E]=[521, 45975]$.

$${\scriptsize
  \renewcommand{\arraystretch}{0.7}
G_{27}^4= \left[ \begin{array}{c}
100000010011100010113320130\\
010000001011010100013333030\\
001000001010101010002131131\\
000100000010010111112303103\\
000010011110010010101330201\\
\hline
002201211111010101100000110\\
202000310001011010110110110\\
\hline
200000000000000000002002002\\
020000000000000000002222222\\
002000000000000000000002220\\
000200000000000000002222002\\
000020000000000000002220220\\
000002000000000000000202202\\
000000200000000000000020222\\
000000020000000000000000202\\
000000002000000000000202020\\
000000000200000000002022022\\
000000000020000000000022020\\
000000000002000000000002200\\
000000000000020000002002020\\
\end{array}\right],~
G_{28}^4= \left[ \begin{array}{c}
1000000001101001010102330111  \\
0100000000110100101011211211  \\
0010000100011010010101123101  \\
0001000010001101001011132130   \\
0000100101000110100100331031  \\
\hline
2002012011010111100010011100  \\
0202221110100100011000101111  \\
\hline
2000000000000000000000020202  \\
0200000000000000000002000022  \\
0020000000000000000002200002  \\
0002000000000000000002200020  \\
0000200000000000000000222000  \\
0000020000000000000000022200  \\
0000002000000000000000002220  \\
0000000200000000000002020222 \\
0000000020000000000002220000   \\
0000000002000000000000220002  \\
0000000000200000000002022000  \\
0000000000020000000000222220  \\
0000000000002000000000020220  \\
0000000000000200000000002022  \\
\end{array}\right]
}
$$

	$$
{\scriptsize
  \renewcommand{\arraystretch}{0.7}
	G_{29}^4= \left[ \begin{array}{c}
10000001110001000001111103030  \\
01000001101100110101002001303   \\
00100001100001001101112030103   \\
00010001101110010010002110213   \\
00001000110111000111010010212   \\
\hline
00222301001010101110111110111  \\
00002030011011100000101111100  \\
\hline
20000000000000000000002022220  \\
02000000000000000000000220022  \\
00200000000000000000000202002  \\
00020000000000000000002200000  \\
00002000000000000000002000222  \\
00000200000000000000002222022  \\
00000020000000000000002200202  \\
00000002000000000000002022222  \\
00000000200000000000002200022  \\
00000000020000000000002222200  \\
00000000002000000000002002202  \\
00000000000200000000000020200  \\
00000000000020000000000222002   \\
00000000000002000000002020202  \\
00000000000000200000000202202 \\
\end{array} \right],
}
$$

  \renewcommand{\arraystretch}{0.7}
{\scriptsize
$G_{33}^4=\left[ \begin{array}{c}
100000000100101011010001020213213  \\
010000000110010100011100203101010  \\
001000000011111010001001010333131  \\
000100000110111000110111200023000  \\
000010000001010111111011113223001  \\
000001000101011111111100312212322  \\
000000100001001111100011302030032  \\
\hline
222022030110011001100000010101011  \\
002200023110101001100111000110000  \\
\hline
200000000000000000000002022022000 \\
020000000000000000000002202022020  \\
002000000000000000000002002022202  \\
000200000000000000000002202020022  \\
000020000000000000000002200022020  \\
000002000000000000000002002022002  \\
000000200000000000000000002000022  \\
000000020000000000000002220022220  \\
000000002000000000000002020002222  \\
000000000200000000000000220022002  \\
000000000020000000000002202000202  \\
000000000002000000000002202022202  \\
000000000000200000000000202022002  \\
000000000000020000000000020002002  \\
000000000000000200000002002022020 \\
\end{array} \right],
 G_{34}^4=\left[ \begin{array}{c}
1000000000010110100101010311033313 \\
0100000000001001110101000113202213  \\
0010000000011010010101003100021323  \\
0001000000000111011101110221331112  \\
0000100000101000101001001031101030  \\
0000010000011111010001100322300012  \\
0000001000110010100111111103122323  \\
0000000100111010010111112213011100   \\
\hline
0002002032101001101101111001101011   \\
2222222201110111011111001111111101  \\
\hline
2000000000000000000000000220020022  \\
0200000000000000000000000002022200  \\
0020000000000000000000020222022220  \\
0002000000000000000000000022220222  \\
0000200000000000000000000000222022   \\
0000020000000000000000020022222200  \\
0000002000000000000000020222002020  \\
0000000200000000000000020202020022   \\
0000000020000000000000020222222202   \\
0000000002000000000000020000200222  \\
0000000000200000000000000202222000 \\
0000000000002000000000000222202022  \\
0000000000000002000000020220002202  \\
0000000000000000002000000200020220 \\
		\end{array} \right].$
}	

		\begin{table}[htbp]
		\caption{Minimum norms and kissing numbers} 
		\label{minimum norms and kissing numbers}
		\centering
		\begin{tabular}{ccc}
			\toprule
		$L$ &	$\mu(L)$ & $N(L)$ \\
			\midrule 
		$A_4(\C_{27}^4)$    &   $3$   & $2664$  \\
		$A_4(\C_{28}^4)$    &   $3$   & $2240$  \\
		$A_4(\C_{29}^4)$    &   $3$   & $1856$  \\
		$A_4(\C_{33}^4)$    &   $3$   & $704$   \\
		$A_4(\C_{34}^4)$    &   $3$   & $544$  \\
			\toprule
		\end{tabular}
	\end{table}
	
	In Table \ref{minimum norms and kissing numbers}, we list the minimum norms $\mu(L)$ and the kissing numbers $N(L)$ of  lattices $L= A_4(\C_n^4)$ constructed from $\C_n^4$ given in Table \ref{new optimal self-dual codes} by Construction $A_4$ in Section 2.
		Note that the kissing numbers in Table \ref{minimum norms and kissing numbers} are different from those in~\cite{Ha} for $n=28, 33, 34$. In~\cite[p. xliv]{Sphere Packings}, it is remarked that there are exactly three Type I lattices in dimension 27 and 38 Type I lattices in dimension 28. For $n=29$ only one Type I lattices in dimension 29 is known~\cite{NebSlo}. Hence our Type I lattices $A_4(\C_{27}^4)$, $A_4(\C_{28}^4)$, and $A_4(\C_{29}^4)$ are alternative lattices in dimensions 27, 28, and 29. Our lattice $A_4(\C_{33}^4)$ have the same kissing number as the one in ~\cite{NebSlo}. However our lattice $A_4(\C_{34}^4)$ seems to be new since the unimodular lattice in dimension 34 denoted by 34MIN3 in \cite{NebSlo} has a kissing number $560$ and this was the only known one from~\cite[Table 2]{Gab2004} by 2004. 

		\section{Conclusion}
		In this paper, we have succeeded in expanding any self-orthogonal code to many self-dual codes over  $\Z_4$. We have explored various methods for expansion over the ring $\Z_4$, focusing on maintaining the dimension of the residue code or increasing it to its maximum capacity for a given free or non free  self-orthogonal $\Z_4$ code. Using Magma~\cite{magma}, we have recovered all indecomposable self-dual codes with lengths from $4$ to $8$ by expanding self-orthogonal codes. When it comes to longer lengths, we constructed five new optimal self-dual codes over $\Z_4$ of length $27, 28, 29, 33, 34$ in Table \ref{new optimal self-dual codes}. 
	Our research not only demonstrates the feasibility of expanding a self-orthogonal code into many self-dual codes over the ring $\Z_4,$ but also provides explicit and comprehensive algorithms. We
	expect that three algorithms will provide more interesting new self-dual codes.

\bigskip

		\begin{appendices}

		\end{appendices}
		


\noindent
{\Large {\bf Appendix}}

		\begin{algorithm}[H]\label{algorithm1}
			\SetAlgoLined
			\KwIn{A self-orthogonal code $\C$ over $\Z_4$.}
			\textbf{Step 1:} Obtain the type $\left\{ k_1,k_2 \right\} $, the length $n$, generator matrix $G$ of $\C$.

			\While{$\C$ is not self-dual}{
				\textbf{Step 2:} Add vector to $\C$ of the form $2\textbf{c}$, where $\textbf{c} \in {\mbox{Res}}(\C)^{\bot} ~\backslash ~{\mbox{Tor}}(\C)$.

				\textbf{Step 3:} Let $G'=\left[ \begin{array}{c}
					G\\
					2\textbf{c}\\
				\end{array} \right] $ and generate a new self-orthogonal code $\C'$.

				\textbf{Step 4:} Let $\C=\C'$.
				
			}
			\KwOut{ A self-dual code of type $\left\{ k_1,n-2k_1 \right\}$ and its generator matrix is $G'$. }
			\caption{How to expand a self-orthogonal code into a self-dual code over $\Z_4$ and keep the value of $k_1$ (based on Theorem \ref{thm-ext-1}) }
		\end{algorithm}

		\begin{algorithm}[H]
			\SetAlgoLined   \label{algorithm2}
			\KwIn{A self-orthogonal code $\C$ over $\Z_4$.}
			\textbf{Step 1:} Obtain the type $k_1 $, the length $n$, generator matrix $G$ of $\C$, and let
			\begin{equation*}
				G=\left[ \begin{matrix}
					I_{k_1}&		A&		B_1+2B_2
				\end{matrix} \right].
			\end{equation*} \\
			\textbf{Step 2:} Find a binary doubly-even code $\mathcal{C}_{de}$ with dimension $k_1'>k_1$ which contains Res$(\C)$.

			\textbf{Step 3:} Write the generator matrix of $\mathcal{C}_{de}$ as
				\begin{equation*}
				G_2=\left[ \begin{matrix}
					I_{k_1}&		A&		B_1\\
					\bf{0}&		I_{k_1'-k_1}&		C\\
				\end{matrix} \right],
			\end{equation*}
			where the first $k_1$ rows form the generator matrix for Res$(\C)$.

			\textbf{Step 4:} Apply Lemma \ref{lem-binquar} to transform $G_2$ into many matrices over $\Z_4$, and choose one suitable matrix $G_3$  in the form of 	\begin{equation*}
			G_3=\left[ \begin{matrix}
					I_{k_1}&		A&		B_1+2B_2\\
					F&		E&		C\\
				\end{matrix} \right].
			\end{equation*} \\
		\textbf{Step 5:} Let $G_3$ generate a self-orthogonal code $\C'$ over $\Z_4$.
		
		\textbf{Step 6:} Apply Algorithm \ref{algorithm1} to obtain a self-dual code $\C_{sd}$ containing $\C'$.
			




				
				\KwOut{A self-dual code with the type  $\{ k_1', n-2k_1'\}$.}
				
			\caption{How to expand a free self-orthogonal code $\C$ in a standard form into many dfifferent self-dual codes over $\Z_4$ (based on Theorem \ref{thm-ext-2} ) }
		\end{algorithm}

		\begin{algorithm}[H]\label{algorithm3}
			\SetAlgoLined
			\KwIn{A self-orthogonal code $\C$ over $\Z_4$.}
			\textbf{Step 1:} Obtain the type $\left\{ k_1,k_2 \right\} $, the length $n$, generator matrix $G$ of $\C$.
				$$
			G=\left[ \begin{matrix}
				I_{k_1}&		A&	      	B_1+2B_2\\
				\bf{0}      &        2I_{k_2}&    2C     \\
			\end{matrix} \right].
			$$\\
			\textbf{Step 2:} Find a binary double-even code $\C_{de}$ with dimension $k_1'>k_1$ such that  ${\text{Res}}(\C) \subseteq \C_{de}$ and $\C_{de} \subseteq {\text{Tor}}(\C)^\perp$, and write the generator matrix of $\C_{de}$ as
			\begin{equation*}
				G_2=\left[ \begin{matrix}
					I_{k_1}&		A&		B_1\\
					\bf{0}&		I_{k_1'-k_1}&		D\\
				\end{matrix} \right].
			\end{equation*}\\
			\textbf{Step 3:} Apply Lemma \ref{lem-binquar} to transform $G_2$ into many  matrices over $\Z_4$ and choose one suitable matrix $G_3$  in the form of 	\begin{equation*}
			G_3=\left[ \begin{matrix}
				I_{k_1}&		A&		B_1+2B_2\\
				F&		E&		C\\
			\end{matrix} \right].
		\end{equation*} \\
			\textbf{Step 4:}   Define $G_4$ as follows to generate a new self-orthogonal code $\C'$ over $\Z_4$:
				\begin{equation*}
				G_4=\left[ \begin{matrix}
					I_{k_1}&		A&		B_1+2B_2\\
					F&		E&		C\\
				 \bf{0}&        2I_{k_2}&    2C     \\
				\end{matrix} \right].
			\end{equation*}\\

\textbf{Step 5:} Apply Algorithm \ref{algorithm1} to obtain a self-dual code $\C_{sd}$ containing $\C'$.


			\KwOut{a self-dual code with the type  $\{ k_1', n-2k_1'\}$. }
								
				\caption{How to expand a self-orthogonal code $\C$ in a standard form into many different  self-dual codes over $\Z_4$ (based on Theorem \ref{thm-ext-3}) }
			\end{algorithm}

		\end{document}